\documentclass[runningheads]{llncs} % nolineno for no line numbers 
\usepackage[utf8]{inputenc}
\usepackage{amsmath,amssymb}
\usepackage[basic]{complexity}
\usepackage{enumitem,xspace}

\usepackage[usenames,dvipsnames,table]{xcolor}
\usepackage{graphicx}
\DeclareGraphicsExtensions{.pdf,.png,.eps}
\graphicspath{{figs/}} 
\definecolor{defblue}{rgb}{0.08235294118,0.3098039216,0.537254902}
\let\emph\relax
\DeclareTextFontCommand{\emph}{\color{defblue}\em}

\usepackage{url}
\usepackage{cite}
\usepackage[raiselinks=true,
			bookmarks=true,
			bookmarksopen=true,
			bookmarksopenlevel=2,
			bookmarksnumbered=true,
			hyperindex=true,
			plainpages=false,
			pdfpagelabels=true,
			pdfborder={0 0 0.5},
			colorlinks=false,						
			linkbordercolor={0 0.61 0.50},   
			citebordercolor={0 0.61 0.50}]{hyperref} 
\usepackage[capitalise,noabbrev,nameinlink ]{cleveref}
\newcommand{\etal}{{et~al.}}
\renewcommand{\orcidID}[1]{\href{https://orcid.org/#1}{\includegraphics[scale=.03]{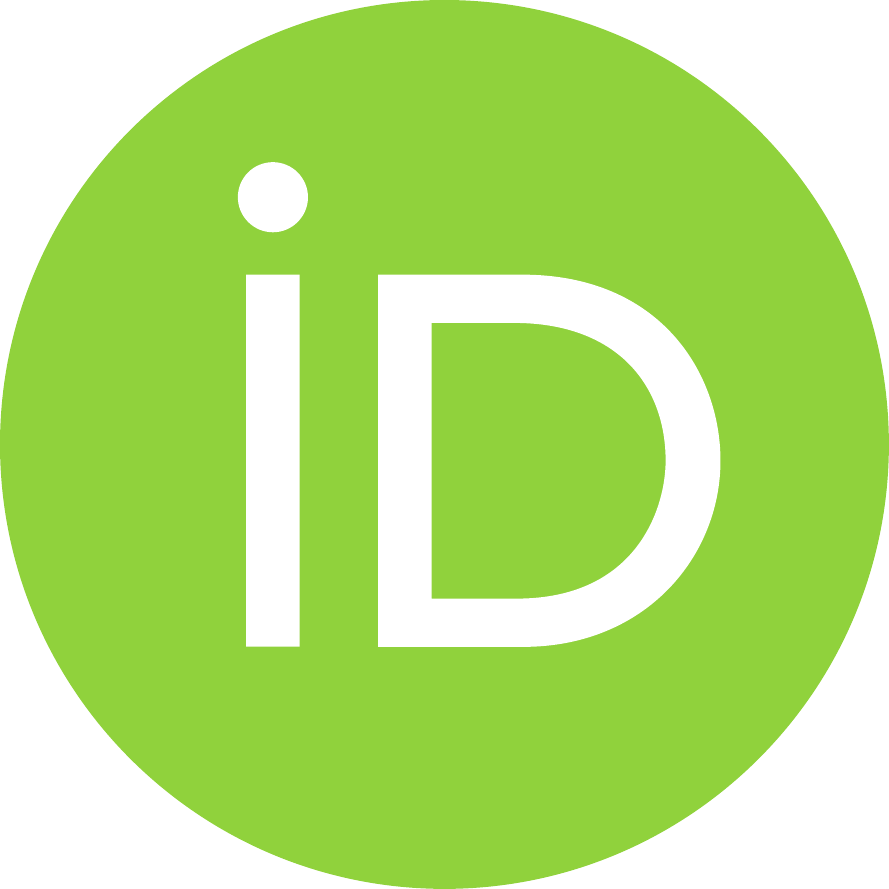}}}
\newcommand{\doii}[1]{\href{https://dx.doi.org/#1}{\texttt{doi:#1}}}

\usepackage{mathtools}
\DeclarePairedDelimiter\set{\{}{\}}
\DeclarePairedDelimiter\abs{\lvert}{\rvert}
\def\to{\ensuremath{\rightarrow}}
\def\Oh{\ensuremath{\mathcal{O}}}
\def\N{\ensuremath{\mathbb{N}}}

\begin{document}
\title{Drawing Tree-Based Phylogenetic Networks with Minimum Number of Crossings}
\titlerunning{Crossing Minimization for Drawings of Tree-Based Networks}
\author{Jonathan Klawitter\inst{1}\orcidID{0000-0001-8917-5269} \and\\
Peter Stumpf\inst{2}\orcidID{0000-0003-0531-9769}}
\authorrunning{Klawitter and Stumpf}
\institute{University of Würzburg, Germany
\and
	University of Passau, Germany
}
\maketitle            

\pdfbookmark[1]{Abstract}{Abstract} 
\begin{abstract}
In phylogenetics, tree-based networks are used to model and visualize 
the evolutionary history of species 
where reticulate events such as horizontal gene transfer have occurred.
Formally, a tree-based network $N$ consists of a phylogenetic tree $T$ 
(a rooted, binary, leaf-labeled tree)
and so-called reticulation edges that span between edges of $T$.
The network $N$ is typically visualized by drawing $T$ downward and planar
and reticulation edges with one of several different styles.  
One aesthetic criteria is to minimize the number of crossings between
tree edges and reticulation edges. 
This optimization problem has not yet been researched.  
We show that, if reticulation edges are drawn x-monotone,
the problem is NP-complete, but fixed-parameter tractable in the number of reticulation edges. 
If, on the other hand, reticulation edges are drawn like ``ears'',
the crossing minimization problem can be solved in quadratic time.
\end{abstract}

\section{Introduction} \label{sec:introduction}
%> phytrees and phynets
The evolution of a set of species 
is usually depicted by a \emph{phylogenetic tree}~\cite{SS03}.
More precisely, a \emph{phylogenetic tree} $T$ is a rooted, binary tree 
where the leaves are labeled bijectively by the set of species. 
The internal vertices of $T$, each having two children, represent bifurcation events 
in the evolution of the taxa.
The heights assigned to vertices indicate the flow of time from the root, 
lying furthest in the past, to the present-day species.

Evolutionary histories can however not always be fully represented by a tree~\cite{Doo99}.
Indeed, reticulate events such as hybridization, horizontal gene transfer, recombination, and reassortment
require the use of vertices with higher indegree~\cite{HRS10,Ste16}.
A \emph{phylogenetic network} $N$ generalizes a phylogenetic tree in exactly this sense, that is,
besides the root, leaves and vertices with indegree one and outdegree two, 
$N$ may contain vertices with indegree two and outdegree one. 

%> tbased
\paragraph{Tree-Based Networks.}
Motivated by the question of whether the evolutionary history of the taxa is fundamentally tree-like,
Francis and Steel~\cite{FS15} introduced a class of phylogenetic networks called \emph{tree-based networks}, 
which are ``merely phylogenetic trees with additional edges''.
Formally, a \emph{tree-based network} $N$ is a phylogenetic network that has a subdivision $T'$ of a phylogenetic tree $T$ as spanning tree.
% The tree $T$ is then called a \emph{base tree} of $N$.
Then $T$ is called the \emph{base tree} of $N$ and $T'$ the \emph{support tree} of $N$.
Lately, tree-based networks have received a lot of attention in combinatorial phylogenetics~\cite{FS15,AAA16,JvI18,PSS19}
and while drawings of several other types of phylogenetics networks have been investigated in the past~\cite{Hus09,HRS10,TK16,CDMP20},
this has, to the best of our knowledge, not been done for tree-based networks.
In this paper, we look at drawings of tree-based networks with different drawing styles 
inspired by drawings in the literature.

% assumptions
For a tree-based network $N$, we assume that both the base tree $T$
and the support tree $T'$ as spanning tree of $N$ are fixed.
% In a slight abuse of notation, we let $T$ denote both the base tree and the spanning tree of~$N$.
We call an edge not contained in the embedding of $T'$ into $N$ a \emph{reticulation edge}.
Therefore, we can perceive a drawing of $N$ as a drawing of $T$ (or $T'$) and the reticulation edges.
A vertex of $N$ that is also in $T$ is called a \emph{tree vertex}.

\paragraph{Drawing styles.}
Our drawing conventions are that $N$ is drawn downwards with vertices at their
fixed associated height and $T$ is drawn planar in the style of a dendrogram,
that is, each tree edge $(u, v)$ consists of a horizontal line segment starting at $u$ and a vertical line segment ending at $v$.
For reticulation edges, we have different drawing styles; see \cref{fig:drawingStyles}.
In the \emph{horizontal style} -- the only style 
where the two endpoints of a reticulation edge must have the same height -- 
reticulation edges are drawn as horizontal line segments. This style has for example been used by Kumar \etal~\cite[Figure 4]{Bears}.
We assume that all horizontal edges come with slightly different heights.  
The next two styles are inspired by Figures 3 and~6 by Vaughan \etal~\cite{Bacteria}. 
There, a reticulation edge $(u, v)$ is drawn with two horizontal and one vertical line segment and thus with two bends.
The styles differ in where the vertical line segment is placed. 
We define vertex $\ell(u,v)$ as follows.
If the lowest common ancestor (lca) $w$ in $T'$ of $u$ and $v$ is a tree vertex,
set $\ell(u,v) = w$.  Otherwise, set $\ell(u,v)$ to be the first
tree vertex below~$w$. 
In the \emph{ear style}, the vertical line segment is placed to the right of
the subtree rooted at $\ell(u, v)$.
In the \emph{snake style}, the vertical line segment lies between $u$ and $v$ and, in particular, 
its x-coordinate lies between the x-coordinates of the left and right subtree of $\ell(u,v)$. 

\begin{figure}
  \centering 
  \includegraphics{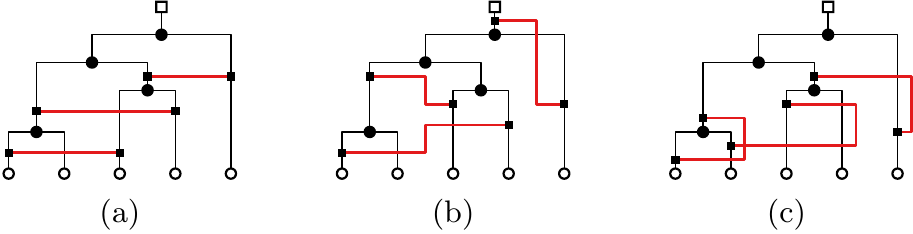}
  \caption{Drawings of tree-based networks with the (a) horizontal, (b) snake,
  and (c) ear style for the red reticulation edges.}
  \label{fig:drawingStyles}
\end{figure}

The aesthetic criteria to optimize for when constructing a drawing of $N$, with either of the styles,
is the number of crossings. Our focus is on crossings between reticulation edges and tree edges. 
Crossings between pairs of reticulation edges may be minimized in a post-processing step.

We make the following important observation.
The number of crossings in a drawing of $N$ is fully determined by the order of the leaves
or, equivalently, by the rotation of each tree vertex. 
Formally,  we use a map $c\colon V(T) \to V(T)$ that assigns to
each non-leaf vertex $v$ of $T$ one of its children.  
In a drawing of $N$, we then consider $v$ to be \emph{rotated left}, if $c(v)$ is its left child,
and \emph{rotated right}, if $c(v)$ is its right child. 
Two vertices are \emph{rotated the same way} if they are both rotated left or if they are both
rotated right. Let $\bar{c}(v)$ denote the child of $v$ that is not~$c(v)$.

\paragraph{Contribution and outline.}
First, we show that the number of crossings can be minimized in quadratic time for ear-style drawings.
Second, we prove that the problem is NP-hard for the horizontal style.
On the positive side, we devise fixed-parameter tractable (fpt) algorithms for the horizontal and the snake style.

%# easy
\section{Ear-Style Drawings: Polynomial-Time Algorithm} % ---------------------------------------
Consider an ear-style drawing of a tree-based network $N$.
Let $e = (u, v)$ be a reticulation edge of $N$ and $f = (x, y)$ a tree edge of $N$.
First, note that the vertical line segment of $e$ is placed such that it does not cross any tree edge.
Next, note that if the subtree $T(\ell(u,v))$ rooted at $\ell(u,v)$ does not contain $f$, then
$e$ and $f$ cannot cross. 
Let $l$ be the horizontal line segment of $e$ starting at~$v$. 
Assume $T(\ell(u,v))$ contains $f$ and the y-coordinate range of $f$ contains the y-coordinate of $v$.
Observe that $l$ and $f$ cross if and only if $f$ is in the right subtree of $\ell(v,y)$; see \cref{fig:ears} (a).
(An analogous condition holds for the horizontal line segment starting at $u$.)
Rotating $\ell(u, v)$ thus changes whether $f$ and $l$ cross. Furthermore, in general, 
the existence of each possible crossing depends on the rotation of a single tree vertex.  
We can thus minimize the number of crossings in an ear-style drawing of $N$ by deciding
for each tree vertex which orientation results in less crossings. 
We show that this can be done efficiently.
  
\begin{theorem}\label{clm:ear}
Let $N$ be a tree-based network with $n$ leaves and $k$ reticulation edges.
Then an ear-style drawing of $N$ with minimum number of crossings can be computed in $\Oh(nk)$ time.
\end{theorem}
\begin{proof}
The idea of the algorithm is to sweep upwards through $N$ and, whenever an endpoint $v$ of a reticulation edge is met, 
to tell $v$'s ancestor tree vertices how many crossings it costs to have $v$ in the left subtree.
Each tree vertex is thus equipped with with two counters that inform about which rotation is less favorable; see \cref{fig:ears} (a).

\begin{figure}[htb]
  \centering
  \includegraphics{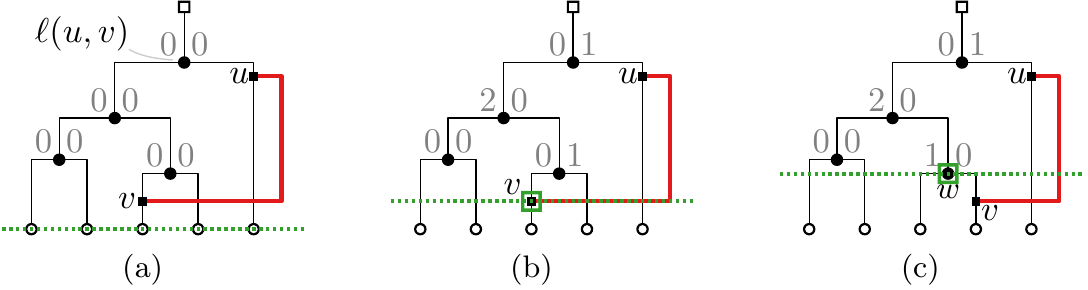}
  \caption{(a) Start of sweep line algorithm with counters at 0; 
  (b) adding potential crossings to counters; 
  (c) rotating $v$ based on counters.} 
  \label{fig:ears}
\end{figure}

Let $e = (u, v)$ be a reticulation edge.
Above we observed that a horizontal segment of $e$ can only have crossings with tree edges below $\ell(u, v)$.  
Therefore, we first compute and store the lca for each pair of endpoints of each reticulation edge in $\Oh(n + k)$ time 
with an algorithm by Gabow and Tarjan~\cite[Section 4.6]{GT85}.
We then start the sweep from the leaves towards the root of $N$. 
At every endpoint $v$ of a reticulation edge $(u, v)$ (or $(v, u)$), 
determine in $\Oh(n)$ time for every vertex~$u$ of $T$ 
the width of its left and right subtree at the height of $v$;
for example with a post-order traversal of $T$.
Then from $v$ up to $\ell(u, v)$,
add for each tree vertex $w$ the width of the subtree not containing $v$ to the respective counter;
see \cref{fig:ears}~(b).
This way, we count potential crossings of the horizontal segment at $v$ with the vertical segments of all edges at the height of $v$ in this subtree at once.
When the sweep reaches a tree vertex $w$, as in \cref{fig:ears}~(c), pick the best rotation for $w$ based on its counters.
In total we have $2k$ steps for endpoints of reticulation edges taking $\Oh(n)$ time and
$\Oh(n)$ steps for tree vertices taking $\Oh(1)$ time. Hence, the algorithm runs in $\Oh(nk)$ time.
\end{proof}

% Note that for $k \geq 2$ and $n > k$, we can easily find a tree-based network $N$ 
% such that every ear-style drawing of $N$ has at least $\Omega(nk)$ crossings. 

To minimize crossings between pairs of reticulation edges in a post-processing step,
we only have to consider pairs of reticulation edges that have the vertical segment to the right of the same subtree and that are nested, 
that is, two reticulation edges $(u, v)$ and $(x, y)$ with $u$ above $x$ and $y$ above $v$.
The vertical segment of $(u, v)$ should then be to the right of the vertical segment of $(x, y)$.

%# hard
\section{Horizontal-Style Drawings: NP-Completeness} % ---------------------------------------
In this section, we show that the crossing minimization decision problem 
for horizontal-style drawings is NP-complete.
We prove the NP-hardness with a reduction from MAX-CUT, which is known to be
NP-complete~\cite{GJ79}.
Recall that in an instance of MAX-CUT we are given a graph $G = (V, E)$ 
and a parame\-ter~$p \in \N$, and have to decide whether there exists a 
bipartition $(A, B)$ of $V$ with at least $p$
edges with one end in $A$ and one end in $B$.

\begin{theorem}
The crossing minimization problem for horizontal-style drawings of a tree-based network is NP-complete.
\end{theorem}
\begin{proof}
Firstly, since we can non-deterministically generate all the drawings of $N$
and count the number of crossings of a drawing in polynomial
time, the problem is in NP. Concerning the hardness, we polynomial-time
reduce a MAX-CUT instance with a graph $G = (V, E)$ to crossing minimization on a
tree-based network $N$. In the following construction of $N$, assume
that leaves are always (re)assigned the height $0$.

The main idea is to have one \emph{edge gadget} $N_e$  for each $e \in E$ that induces
a crossing if and only if $e$ is not in our cut; see \cref{fig:gadgets}. 
Let $h\colon V \to \mathbb{N}$ be an arbitrary vertex ordering.
Let $e = \set{u, v} \in E$ and suppose $h(u) < h(v)$. 
The construction of $N_e$ then works as follows.
We have a tree vertex $u_e$ with two leaves as children and a tree vertex $v_e$
with $u_e$ and a leaf as children. We set~$c(v_e) = u_e$
and the heights of $u_e$ and $v_e$ to $h(u)$ and $h(v)$ respectively.
We add a reticulation edge $f_e$ between $u_ec(u_e)$ and $v_e\bar{c}(v_e)$.
Note that $f_e$ and $u_e\bar{c}(u_e)$ cross if and only if $u_e$ and $v_e$ are
rotated the same way.
To connect all edge gadgets, we replace the leaves of an arbitrary rooted, binary tree
with $\abs{E}$ leaves and a downward planar embedding with the edge gadgets; 
see \cref{fig:reduction}.

\begin{figure}[htb]
  \centering 
  \includegraphics{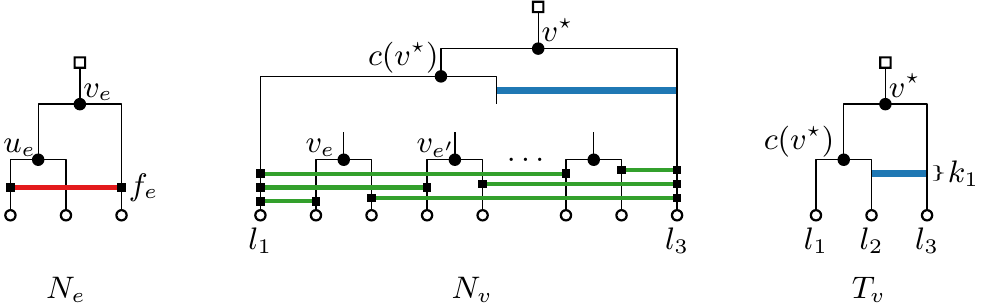}
  \caption{An edge gadget $N_e$; a vertex gadget $N_v$ based on the tree $T_v$.}
  \label{fig:gadgets}
\end{figure}

We want to ensure that the tree vertices $v_1,\dots,v_{\deg(v)}$ corresponding
to the same node $v\in V$ are all rotated the same way. 
If this is enforced, we can consider all nodes in $V$ where the corresponding tree vertices
are rotated left as one partition set and all nodes in $V$ where the
corresponding tree vertices are rotated right as the other partition set.
If on the other hand a cut is given, we simply choose for each vertex the
rotation of the corresponding tree vertices accordingly.
Now, to ensure the same rotation for all corresponding tree vertices, we construct a
\emph{vertex gadget} $N_v$ for each node $v \in V$ (in some order); see \cref{fig:gadgets}.
We start with a rooted, binary tree $T_v$ on three leaves
$l_1, l_2, l_3$ such that $l_1$ and $l_2$ have a common parent.
Let $v^\star$ denote the child of the root of $T_v$ and let $c(c(v^\star)) = l_1$.
Add a bundle of $k_1 = 2 (\abs{V} + 1) \cdot \abs{E}$ reticulation edges between $l_2$ and $l_3$.
We will see that $k_1$ is large enough such that this bundle does not induce crossings in a
crossing minimum drawing. It thus enforces that $l_2$ lies between $l_1$ and $l_3$. %,
%orders $l_1$, $l_2$, $l_3$ (up to reversal)
%and prevents the tree vertices of $T_v$ from rotating independently.
We substitute $l_2$ by our current construction; see \cref{fig:gadgets}.

Lastly, for $1 \le i \le \deg(v)$, we add a 
reticulation edge between $v_ic(v_i)$ and the incoming edge of $l_1$, 
and a reticulation edge between $v_i\bar{c}(v_i)$ and $l_3$. 
Note that if $v^\star$ and $v_i$ are rotated the same way, we
get two crossings less than otherwise.
However, different rotations can save at most one crossing in the edge gadget containing $v_i$.
Hence, in a crossing minimum drawing, $v^\star$ and $v_i$ are rotated the same
way. In fact, $v_1,\dots,v_{\deg(v)}, v^\star$ are rotated the same way.
This completes the construction of $N$. Note that $N$ has a size polynomial 
in the size of $G$.

Note that the order of the edge gadgets does not
influence the number of crossings with the two reticulation edges added for $v_i$; 
this number is fixed for crossing minimum drawings. 
Therefore, we can compute the total number $k_2$ of crossings induced by vertex gadgets.
Furthermore, since $k_2 \le 2\abs{V} \ \abs{E}$ and thus~$k_1 \ge k_2 + \abs{E} + 1$, 
we get that crossing one edge bundle would induce more crossings than we obtain from the vertex gadgets 
and from the edge gadgets. Hence, no bundle induces crossings in a crossing minimum drawing. 

\begin{figure}[htb]
  \centering 
  \includegraphics{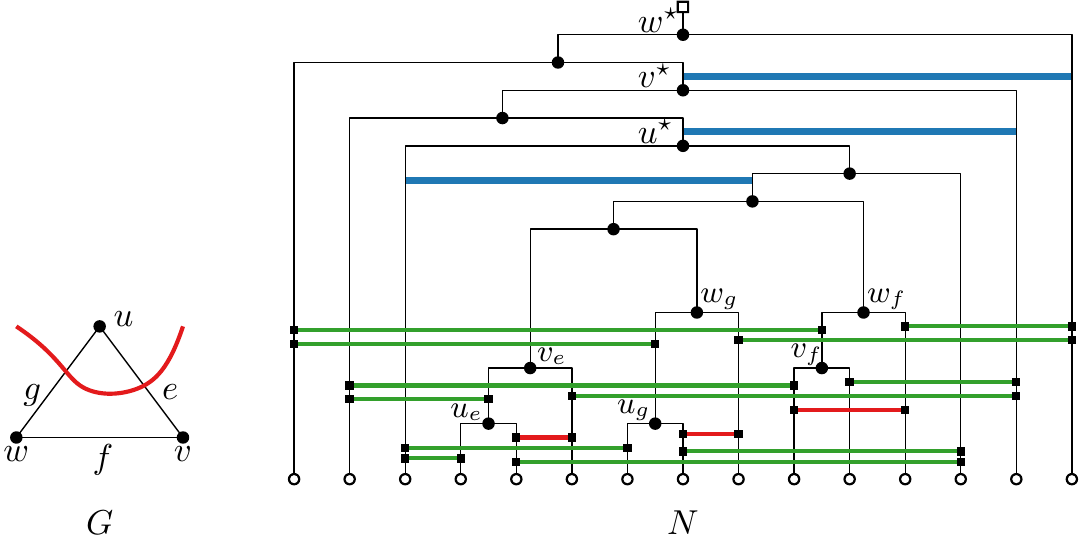} 
  \caption{A crossing-minimum drawing of $N$ inducing a max-cut on $G$.}
  \label{fig:reduction}
\end{figure}

We conclude that minimizing crossings boils down to minimizing crossings in edge gadgets.  
Finally, by the construction of $N$ and our observations, we get that
$N$ admits a horizontal-style drawing with $k \le k_2 + \abs{E} - p$ crossings 
if and only if $G$ admits a cut of size at least $p$. The statement follows.
\end{proof}

A snake-style drawing where endpoints of reticulation edges
have the same height is a horizontal-style drawing; the reduction thus also works for this style.

\begin{corollary} 
The crossing minimization problem for snake-style drawings of a tree-based network is NP-complete.
\end{corollary}

\section{Snake-Style Drawings: FPT Algorithm} % ---------------------------------------
For the ear style, we have seen that whether a reticulation edge and a tree
edge cross, depends on the rotation of at most one tree vertex,
since horizontal line segments always go to the right.
This is not the case for horizontal-style and snake-style drawings.  
However, fixing the rotation of $\ell(u, v)$ for each reticulation edge $(u, v)$, also fixes
for the horizontal line segments of $(u, v)$ whether they go to the left or right. 
Further, while the vertical line segment may have a single crossing, this
crossing occurs if and only if one endpoint of the reticulation edge is the
lca of both endpoints. 
We can again conclude that the existence of each crossing of a horizontal line
segment with a tree edge depends on the rotation of a single tree vertex --
with two differences to the ear style: (i) A horizontal line segment can now also go towards the
left. (ii) A horizontal line segment of a reticulation edge $(u,v)$ ends
between the two subtrees of $\ell(u,v)$, i.e., one of the two subtrees can have
crossings with only one of the horizontal line segments of $(u,v)$.
With these observations we can now devise a fixed-parameter tractable algorithm. 

\begin{theorem} \label{clm:fpt}
Let $N$ be a tree-based network with $n$ leaves and $k$ reticulation edges.
Then a snake-style drawing of $N$ with minimum number of crossings can
be computed in $\Oh(2^{k}\cdot nk)$ time. The computation is thus
fixed-parameter tractable when parametrized by $k$.
\end{theorem}
\begin{proof}
  Let $L = \{\ell(u,v)\mid (u,v)$ is a reticulation edge$\}$.  
  Suppose the rotation for all $v \in L$ is fixed.
  With the observation above, we can slightly adapt our algorithm from \cref{clm:ear} 
  to compute for every $v \not\in L$ the rotation that induces less crossings.
  Namely, the algorithm has to differentiate whether line segments go to the left or right,
  and pick a rotation only for $v \not \in L$.  

  We try this for all possible combinations of rotations of vertices in $L$ 
  and then pick the drawing with the least crossings. 
  Since there are $\Oh(2^{k})$ such combinations, the statement on the running-time follows.
\end{proof}

Note that this implies the same statement for the horizontal style.

\pdfbookmark[1]{References}{References}  
% \bibliographystyle{splncs04}
% \bibliography{sources}  

\end{document}